\documentclass[final,journal]{IEEEtran}

\usepackage{url}
\usepackage{cite}
\usepackage[cmex10]{amsmath} 
\usepackage{amssymb,amsthm}
\usepackage{dsfont}
\usepackage{mathtools}
\usepackage{amsmath}
\usepackage{xfrac}
\usepackage{stmaryrd}
\usepackage{tikz}
\usepackage{centernot}
\usepackage[letterspace=150]{microtype}
\usepackage{cryptocode}

\makeatletter
\g@addto@macro{\UrlBreaks}{\UrlOrds}
\makeatother

\makeatletter\let\glb@currsize\@empty\makeatother

\usetikzlibrary{automata, arrows}

\newtheorem{theorem}{Theorem}

\newtheorem{definition}{Definition}
\newtheorem{proposition}[theorem]{Proposition}

\newtheorem{remark}{Remark}

\newtheorem{corollary}[theorem]{Corollary}

\newcommand{\range}[1]{\llbracket #1 \rrbracket}

\begin{document}
\title{Optimal Pre-Processing to Achieve Fairness and Its Relationship with Total Variation Barycenter} 

\author{Farhad Farokhi
	\thanks{F. Farokhi is with the Department of Electrical and Electronic Engineering at the University of Melbourne. Email: farhad.farokhi@unimelb.edu.au}\thanks{The work of F. Farokhi is funded by a startup grant from the Melbourne School of Engineering at the University of Melbourne.}}

\maketitle

\begin{abstract} We use disparate impact, i.e., the extent that the probability of observing an output depends on protected attributes such as race and gender, to measure fairness. We prove that disparate impact is upper bounded by the total variation distance between the distribution of the inputs given the protected attributes. We then use pre-processing, also known as data repair, to enforce fairness. We show that utility degradation, i.e., the extent that the success of a forecasting model changes by pre-processing the data, is upper bounded by the total variation distance between the distribution of the data before and after pre-processing. Hence, the problem of finding the optimal pre-processing regiment for enforcing fairness can be cast as minimizing total variations distance between the distribution of the data before and after pre-processing subject to a constraint on the total variation distance between the distribution of the inputs given protected attributes. This problem is a linear program that can be efficiently solved. We show that this problem is intimately related to finding the barycenter (i.e., center of mass) of two distributions when distances in the probability space are measured by total variation distance. We also investigate the effect of differential privacy on fairness using the proposed the total variation distances. We demonstrate the results using numerical experimentation with a practice dataset.
\end{abstract}

\section{Introduction}
Meteoric rise of machine learning has promoted development of wide-ranging applications based on data~\cite{jiang2017artificial, kourou2015machine,dua2016data}. Automated decision making is getting used in sensitive areas, such as criminal justice system~\cite{doi1011770049124118782533}, hiring policies~\cite{amazon_hiring}, and lending~\cite{lending_nab}, despite its lasting impact on human life. These models are however most often trained/designed based on data samples that can be biased. For instance, policing and criminal justice data  are often negatively biased against African-Americans and models trained based on this data can demonstrate unjustifiably high risk of recidivism~\cite{propublica}. This motivates investigating fairness. 

There are two main approaches to fair decision making: pre-processing and post-processing. In post-processing, a machine learning model, e.g., classifier, is changed so that its output is not correlated with the protected attribute, e.g., race or gender~\cite{zafar2017fairness, bechavod2017penalizing,donini2018empirical}. In pre-processing, the input data is modified so that the protected attribute cannot be predicted~\cite{feldman2015certifying,johndrow2019algorithm}. This means any model trained on the processed data will be fair. Pre-processing is utilized in this paper.

We start with defining statistical parity as a notion of fairness. Statistical parity requires that the probability of observing any outcome from the model is the same across all protected attributes, i.e., the outcome of the model is statistically independent of the protected attribute. This definition has been utilized to develop disparate impact as a measure of fairness~\cite{feldman2015certifying}. We prove that disparate impact is upper bounded by the total variation distance between the distribution of the inputs given the protected attributes. We propose a measure of utility degradation, i.e., the extent to which the success of any model is changed by pre-processing the data. We prove that the utility degradation is upper bounded by the total variation distance between the distribution of the data before and after pre-processing. Therefore, to achieve statistical parity by pre-processing while minimizing the utility degradation, we must minimize total variations distance between the distribution of the data before and after pre-processing subject to a constraint on the total variation distance between the distribution of the inputs given the protected attributes. This problem is a linear program and can be  solved efficiently even when very large~\cite{POTRA2000281}. We show that the problem of optimal pre-processing for achieving fairness is intimately related to  finding the \textit{barycenter (i.e., the center of mass) of the distributions of the data conditioned on the sensitive attribute} with distances among the distributions being measured using total variation distance. In this paper, this is referred to as the \textit{total variation barycenter}, which is a linear program. To connect this paper with the existing literature on privacy and fairness, we also investigate the effect of differential privacy on disparate impact and utility degradation using the developed upper bounds based on total variation distance.

Total variation barycenter is similar to the Wasserstein barycenter~\cite{agueh2011barycenters} with the exception of using total variation distance over the space of probability distributions. The change of distance to total variation results in significant computational improvement. The Wasserstein barycenter problem gives rise to complex optimization problems that are not as computationally efficient as the linear programs formulated in this paper. It is worth mentioning that, in the past, the Wasserstein distance has been used to bound disparate impact and utility degradation~\cite{gordaliza2019obtaining, jiang2020wasserstein}, which has resulted in application of the Wasserstein barycenter to fairness. The methods in this paper are much more computationally friendly. Also, total variation distance was proposed for enforcing fairness in~\cite{calmon2017optimized}. However, that paper only focuses on computational aspects of the problem and the relationship between total variation distance, utility degradation, and disparate impact was not investigated. This paper fills this gap. Also, that paper did not investigate the impact of differential privacy in fairness using total variation distance. The closest study to this paper is~\cite{dwork2012fairness}, which uses total variation distance for measuring fairness. However, in that paper, the distribution from one sensitive group is mapped to the other group while, in this paper, we map both distributions to their barycenter to find the best compromise; see (16) in~\cite{dwork2012fairness} which maps the distribution of one sensitive group to another. 

\section{Fairness}
Consider the problem of forecasting output $Y\in\mathbb{Y}$ using correlated random variable $X\in\mathbb{X}$. The population is divided into two\footnote{Extension to more subgroups follows the same line of reasoning.}
subgroups or categories that represent social divide. This is modeled by random variable $S\in\mathbb{S}:=\{0,1\}$. In the algorithmic fairness literature, variable $S$ is identified as the protected attribute, e.g., $S=0$  for a minority or underrepresented group and $S=1$ for the default group. Model $\mathcal{M}:\mathbb{X}\rightarrow \mathbb{Y}$ achieves statistical parity if 
\begin{align} \label{eqn:statistical_parity}
\mathbb{P}\{\mathcal{M}(X)=y\,&|\,S=1\}\nonumber\\
&=\mathbb{P}\{\mathcal{M}(X)=y\,|\,S=0\},\quad\forall y\in\mathbb{Y}.
\end{align}
If~\eqref{eqn:statistical_parity} holds,  the probability of observing an outcome is the same across both minority and default groups, i.e., $\mathcal{M}(X)$ and $S$ are statistically independent. This condition is generally not satisfied in empirical datasets and common models~\cite{dwork2012fairness}. This could be because of a plethora of reasons: bias in sampling, inherent societal bias, or development of unfair forecasting models. One way to achieve statistical parity is to change/transform the data $X$ to break the dependence/correlation between the data $X$ and the protected attributed $S$. Doing so, we can ensure that every classifier acting on the data is fair with respect to $S$ in the sense of~\eqref{eqn:statistical_parity}. This is called data repair/pre-processing~\cite{dwork2012fairness,gordaliza2019obtaining, jiang2020wasserstein}. To this aim, we define new random variable $\tilde{X}\in\mathbb{X}$ and find conditional probability $P_{\tilde{X}|X,S}(\tilde{x}|x,s)=\mathbb{E}\{\tilde{X}=\tilde{x}|X=x,S=s\}$ such that
\begin{align} 
\sum_{x\in\mathbb{X}} P_{\tilde{X}|X,S}(&\tilde{x}|x,1)\mathbb{P}\{X=x|S=1\}\nonumber\\
&=\sum_{x\in\mathbb{X}} P_{\tilde{X}|X,S}(\tilde{x}|x,0)\mathbb{P}\{X=x|S=0\}.\label{eqn:total_repair}
\end{align}
This equality implies that $\tilde{X}$ and $S$ are statistically independent:
\begin{align*}
\mathbb{P}\{\tilde{X}=\tilde{x}\,&|\,S=1\}=\mathbb{P}\{\tilde{X}=\tilde{x}\,|\,S=0\},\quad\forall \tilde{x}\in\mathbb{X}.
\end{align*}
 If~\eqref{eqn:total_repair} holds, for any forecasting model $\mathcal{M}:\mathbb{X}\rightarrow \mathbb{Y}$, we get 
\begin{align} 
\mathbb{P}\{\mathcal{M}(\tilde{X})=y\,&|\,S=1\}\nonumber\\
&=\mathbb{P}\{\mathcal{M}(\tilde{X})=y\,|\,S=0\},\quad\forall y\in\mathbb{Y}.
\label{eqn:statistical_parity_data_repair}
\end{align}
Therefore, after data repair, we always achieve statistical parity irrespective of what model $\mathcal{M}$ we use. 

%The strict fairness conditions in~\eqref{eqn:statistical_parity} or~\eqref{eqn:statistical_parity_data_repair} might not be always achievable without considerable performance degradation. 
We might be interested in relaxing strict fairness conditions in~\eqref{eqn:statistical_parity} or~\eqref{eqn:statistical_parity_data_repair}  in order to attain fairness approximately. In that case, a useful notion of fairness is 
%\begin{align*}
%\mathfrak{F}(\mathcal{M}):=\min_{y\in\mathbb{Y}}\Bigg\{&
%\frac{\mathbb{P}\{\mathcal{M}(X)=y\,|\,S=1\}}{\mathbb{P}\{\mathcal{M}(X)=y\,|\,S=0\}},\\
%&\frac{\mathbb{P}\{\mathcal{M}(X)=y\,|\,S=0\}}{\mathbb{P}\{\mathcal{M}(X)=y\,|\,S=1\}}\Bigg\}.
%\end{align*}
\begin{align*}
\mathfrak{F}(\mathcal{M}):=\max_{y\in\mathbb{Y}}&\Big|
\mathbb{P}\{\mathcal{M}(\tilde{X})=y\,|\,S=1\}\\
&-\mathbb{P}\{\mathcal{M}(\tilde{X})=y\,|\,S=0\}\Big|.
\end{align*}
This measure is often referred to as disparate impact~\cite{feldman2015certifying}. Note that statistical parity is achieved if and only if $\mathfrak{F}(\mathcal{M})=0$. In the remainder of this paper, we use the notation: 
\begin{align*}
\tilde{Q}_s(\tilde{x})&=\sum_{x\in\mathbb{X}} P_{\tilde{X}|X,S}(\tilde{x}|x,s) \mathbb{P}\{X=x|S=s\}.
\end{align*}
Here, $\tilde{Q}_s(\tilde{x})$ denotes the probability of the event $\tilde{X}=x$ conditioned on the realization of the protected attributed $S=s$.

\begin{theorem} $
\mathfrak{F}(\mathcal{M})\leq 2d_{\mathrm{TV}}(\tilde{Q}_1,\tilde{Q}_0).$
\end{theorem}

\begin{proof}
Note that
\begin{align*}
|
\mathbb{P}\{\mathcal{M}&(\tilde{X})=y\,|\,S=1\}-\mathbb{P}\{\mathcal{M}(\tilde{X})=y\,|\,S=0\}|\\
&\leq |
\mathbb{E}\{\mathds{1}_{\mathcal{M}(\tilde{X})=y}\,|\,S=1\}-\mathbb{P}\{\mathds{1}_{\mathcal{M}(\tilde{X})=y}\,|\,S=0\}|\\
&= \Bigg|\sum_{\tilde{x}\in\mathbb{X}}\mathds{1}_{\mathcal{M}(\tilde{x})=y} \tilde{Q}_1(\tilde{x})-\sum_{\tilde{x}\in\mathbb{X}}\mathds{1}_{\mathcal{M}(\tilde{x})=y}\tilde{Q}_0(\tilde{x})\Bigg|\\
&\leq \sum_{\tilde{x}\in\mathbb{X}} \Big|\mathds{1}_{\mathcal{M}(\tilde{x})=y} \tilde{Q}_1(\tilde{x})-\mathds{1}_{\mathcal{M}(\tilde{x})=y}\tilde{Q}_0(\tilde{x})\Big|\\
&\leq \sum_{\tilde{x}\in\mathbb{X}} |\tilde{Q}_s(\tilde{x})-\tilde{Q}_s(\tilde{x})|
\\
&=2d_{\mathrm{TV}}(\tilde{Q}_1,\tilde{Q}_0).
\end{align*}
This concludes the proof.
\end{proof}

\begin{definition}[Differential Privacy] Assume that $S-X-\tilde{X}$ forms a Markov chain, i.e., $P_{\tilde{X}|X,S}(\tilde{x}|x,s)=P_{\tilde{X}|X}(\tilde{x}|x),s\in\{0,1\}$. We say $P_{\tilde{X}|X}(\tilde{x}|x)$ is $\epsilon$-differentially private ($\epsilon$-DP) if 
\begin{align*}
{P_{\tilde{X}|X}(\tilde{x}|x)}{}\leq \exp(\epsilon)P_{\tilde{X}|X}(\tilde{x}|x'),\quad \forall x,x'\in\mathbb{X}.
\end{align*}
\end{definition}

\begin{corollary} \label{cor:1}
Assume that $P_{\tilde{X}|X}$ is $\epsilon$-differentially private. Then,
$d_{\mathrm{TV}}(\tilde{Q}_1,\tilde{Q}_0)
\leq 1-\exp(-\epsilon).$
\end{corollary}

\begin{proof}
The analytical expression for the Dobrushin  ergodicity  coefficient~\cite{Dobrushin1956} and~\cite[Equation~(29)]{polyanskiy2015dissipation} implies that
$
d_{\mathrm{TV}}(\tilde{Q}_1,\tilde{Q}_0)
\leq 1-\sum_{\tilde{x}\in\mathbb{X}} \min_{x\in\range{X}} P_{\tilde{X}|X}(\tilde{x}|x).
$
Therefore, 
\begin{align*}
d_{\mathrm{TV}}(\tilde{Q}_1,\tilde{Q}_0)
&\leq 1-\sum_{\tilde{x}\in\mathbb{X}} \min_{x\in\range{X}} P_{\tilde{X}|X}(\tilde{x}|x)\\
&\leq 1-\sum_{\tilde{x}\in\mathbb{X}} \min_{x\in\range{X}} \exp(-\epsilon) P_{\tilde{X}|X}(\tilde{x}|x')\\
&=1-\exp(-\epsilon)\sum_{\tilde{x}\in\mathbb{X}} P_{\tilde{X}|X}(\tilde{x}|x')\\
&=1-\exp(-\epsilon).
\end{align*}
This concludes the proof.
\end{proof}

\begin{figure}
	\centering
	\begin{tikzpicture}
	\node[] at (0,0) {\includegraphics[width=.95\columnwidth]{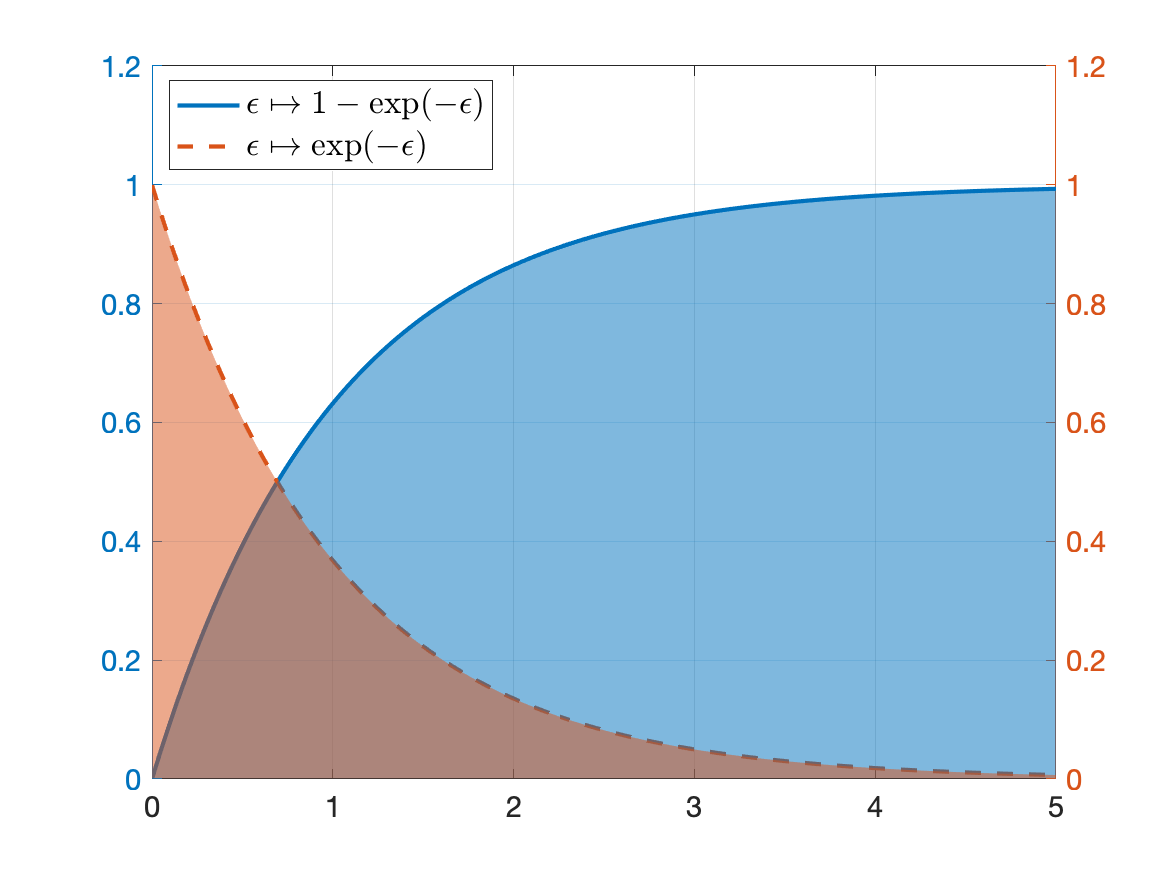}};
	\node[] at (0,-3.1) {\small privacy budget $\epsilon$};
	\node[rotate=90] at (-3.9,0) {\small  $d_{\mathrm{TV}}(\tilde{Q}_1,\tilde{Q}_0)\leq 1-\exp(-\epsilon)$};
	\node[rotate=90] at (+4.3,0) {\small  $\displaystyle\inf_{P_{\tilde{X}|X}\in \mathcal{P}(\epsilon)}  \mathbb{E}\{d_{\mathrm{TV}}(\tilde{P}_S,P_S)\}
		\leq \exp(-\epsilon)$};		
	\end{tikzpicture}
	\caption{
		\label{fig:1}
		Relationship between the upper bound for the disparate impact [solid] and utility degradation [dashed] versus the differential privacy budget $\epsilon$. Differential privacy is effective in data repair to achieve fairness in small privacy budget regime $\epsilon\ll 1$, albeit at the expense of utility.
	}
\end{figure}
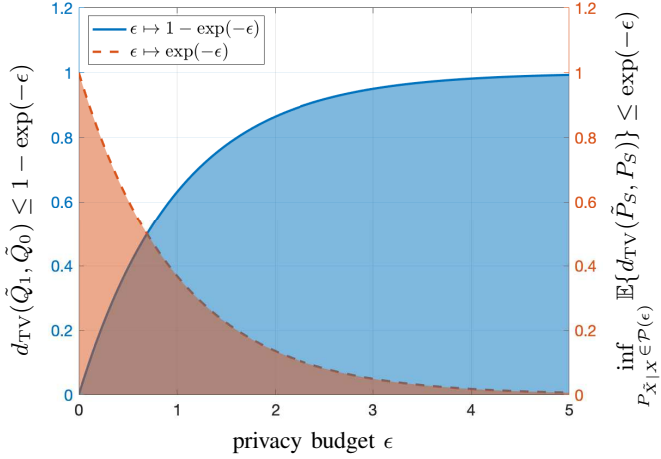

Corollary~\ref{cor:1} shows that, by reducing the privacy budget $\epsilon$, i.e., improving the privacy level, a higher fairness in the sense of disparate impact can be achieved. 

Intuitively, we can achieve fairness if we cannot infer an individual's sensitive/protected attribute from the available data. This way, there is no way to construct unfair models. This is investigated in the following proposition. 

\begin{proposition}
Assume that an adversary aims to estimate $S$ from $\tilde{X}$ using function $\psi:\range{X}\rightarrow \{0,1\}$. Then,
\begin{align*}
\inf_{\psi}\max_{s\in\{0,1\}}\mathbb{P}\{\psi(\tilde{X})\neq S|S=s\} \geq \frac{1}{2}(1-d_{\mathrm{TV}}(\tilde{Q}_1,\tilde{Q}_0))
\end{align*}
\end{proposition}

\begin{proof}
The proof follows from~\cite[Thereom~2.2~(\textit{i})]{tsybakov2008introduction}.
\end{proof}

\section{Utility}
The fairness often comes at the price of utility, i.e., we might not be able to forecast $Y$ as accurately after enforcing fairness. Before data repair/pre-processing, the utility is given by $\mathbb{E}\{\ell(\mathcal{M}(X),Y)\}$, where $\ell(\mathcal{M}(X),Y)$ is measure of closeness between the forecast $\mathcal{M}(X)$ and the actual output $Y$. After data repair, the utility is given by $\mathbb{E}\{\ell(\mathcal{M}(\tilde{X}),Y)\}$. Therefore, the utility degradation is captured by the difference between $\mathbb{E}\{\ell(\mathcal{M}(X),Y)\}$ and $\mathbb{E}\{\ell(\mathcal{M}(\tilde{X}),Y)\}$, i.e.,
\begin{align*}
\mathfrak{U}(\mathcal{M},P_{\tilde{X}|X,S}):=|\mathbb{E}\{\ell(\mathcal{M}(\tilde{X}),Y)\}\hspace{-.03in}-\hspace{-.03in}\mathbb{E}\{\ell(\mathcal{M}(X),Y)\}|.
\end{align*}
Before bounding the utility degradation in the next theorem, we need to define the following notation. In the remainder of this paper, we use the notations:
\begin{align*}
\tilde{P}_s(\tilde{x},y)&=\sum_{x\in\mathbb{X}} P_{\tilde{X}|X,S}(\tilde{x}|x,s) \mathbb{P}\{X=x,Y=y|S=s\},
\\
P_s(x,y)&=\mathbb{P}\{X=x,Y=y|S=s\}.
\end{align*}
Note that $\sum_{y\in\range{Y}}\tilde{P}_s(\tilde{x},y)=\tilde{Q}_s(\tilde{x})$. 

\begin{theorem} Assume that $\ell:\mathbb{Y}\times\mathbb{Y}\rightarrow [-\sigma,\sigma]$ for bounded $\sigma>0.$ Then,
$
\mathfrak{U}(\mathcal{M},P_{\tilde{X}|X,S})	\leq 2\sigma \mathbb{E}\{d_{\mathrm{TV}}(\tilde{P}_S,P_S)\}.
$
\end{theorem}

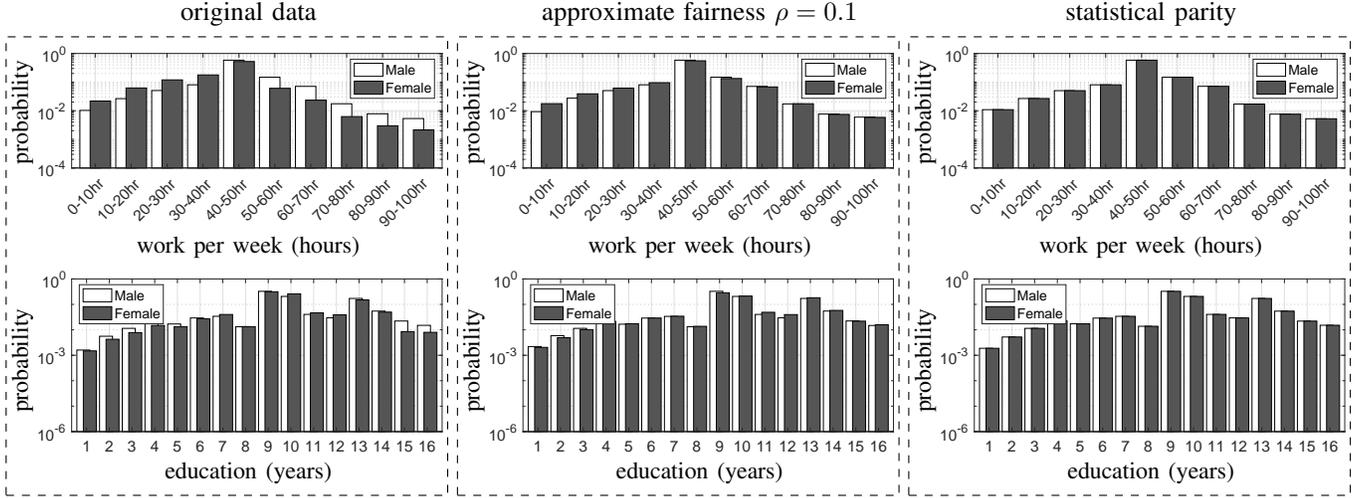
\begin{figure*}[t!]
	\centering
	\begin{tikzpicture}
	\node[] at (0,0) {\includegraphics[width=.35\linewidth]{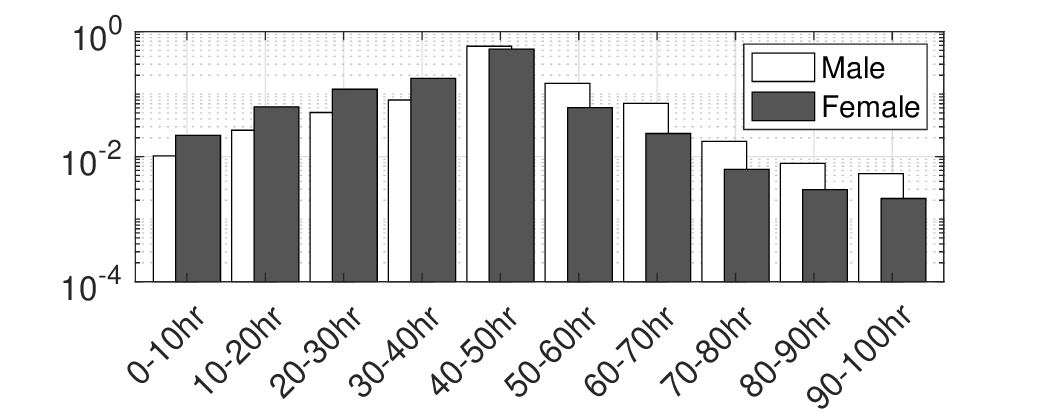}};
	\node[] at (0,-1.5) {\small work per week (hours)};
	\node[rotate=90] at (-3.0,.3) {\small  probability};
	\node[] at (6,0) {\includegraphics[width=.35\linewidth]{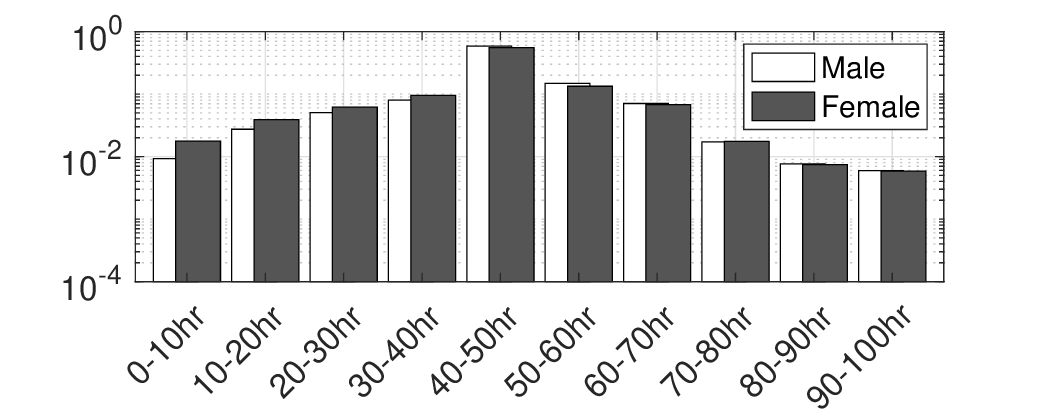}};
	\node[] at (6,-1.5) {\small work per week (hours)};
	\node[rotate=90] at (3,.3) {\small  probability};
	\node[] at (12,0) {\includegraphics[width=.35\linewidth]{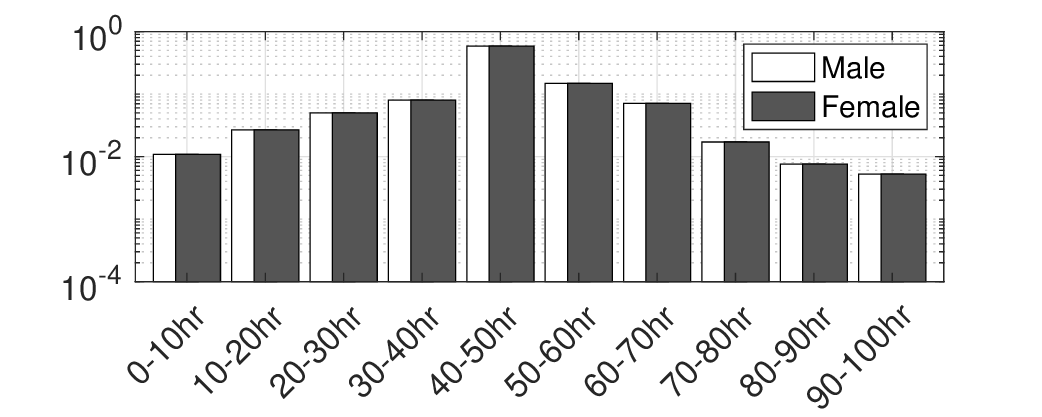}};	
	\node[] at (12,-1.5) {\small work per week (hours)};
	\node[rotate=90] at (9,.3) {\small  probability};	
	\node[] at (0,-3) {\includegraphics[width=.35\linewidth]{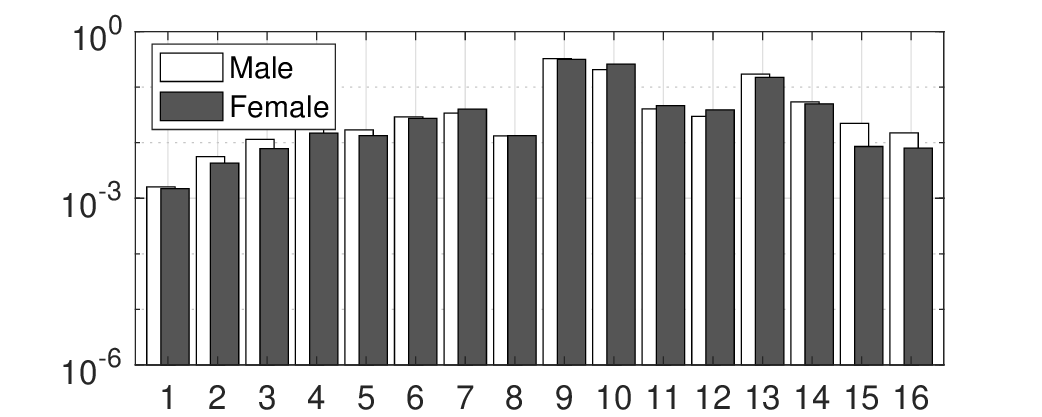}};	
	\node[rotate=90] at (-3,-3) {\small  probability};	
	\node[] at (0,-4.5) {\small education (years)};	
	\node[] at (6,-3) {\includegraphics[width=.35\linewidth]{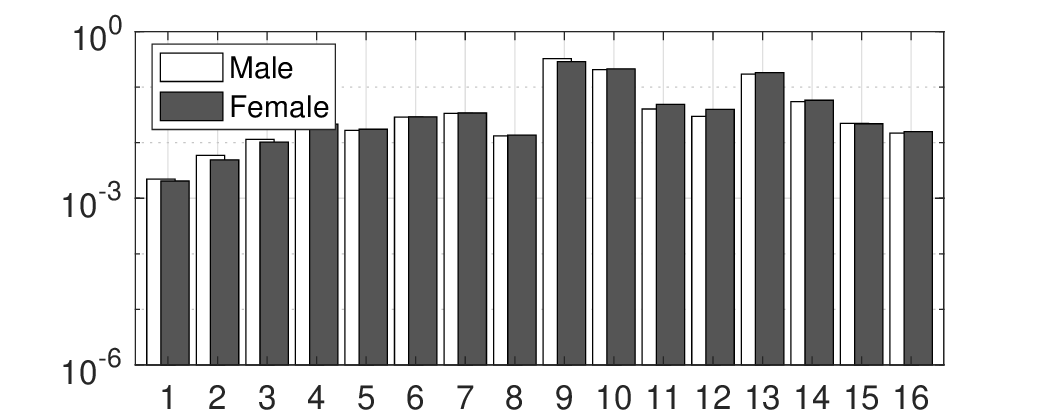}};	
	\node[rotate=90] at (3,-3) {\small  probability};	
	\node[] at (6,-4.5) {\small education (years)};		
	\node[] at (12,-3) {\includegraphics[width=.35\linewidth]{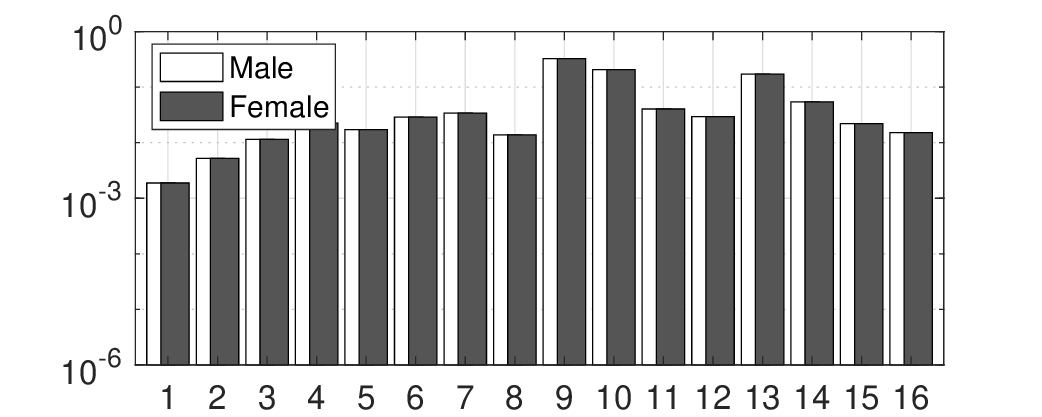}};			
	\node[rotate=90] at (9,-3) {\small  probability};	
	\node[] at (12,-4.5) {\small education (years)};	
	\draw[-,dashed] (-3.22,-4.8)	-- (2.65,-4.8) -- (2.65,1.3) -- (-3.22,1.3) -- cycle;
	\draw[xshift=6cm,-,dashed] (-3.22,-4.8)	-- (2.65,-4.8) -- (2.65,1.3) -- (-3.22,1.3) -- cycle;
	\draw[xshift=12cm,-,dashed] (-3.22,-4.8)	-- (2.65,-4.8) -- (2.65,1.3) -- (-3.22,1.3) -- cycle;	
	\node[] at (0,1.6) {original data};
	\node[] at (6,1.6) {approximate fairness $\rho=0.1$};
	\node[] at (12,1.6) {statistical parity};		
	\end{tikzpicture}
	\caption{
		\label{fig:3}
		Histogram of work per week and education for male and female individuals for the original data [left], pre-processed data with approximate fairness of level $\rho=0.1$ [middle], and pre-processed data to ensure statistical parity [right]. 
	}
\end{figure*}

\begin{figure}
	\centering
	\begin{tikzpicture}
	\node[] at (0,0) {\includegraphics[width=.95\columnwidth]{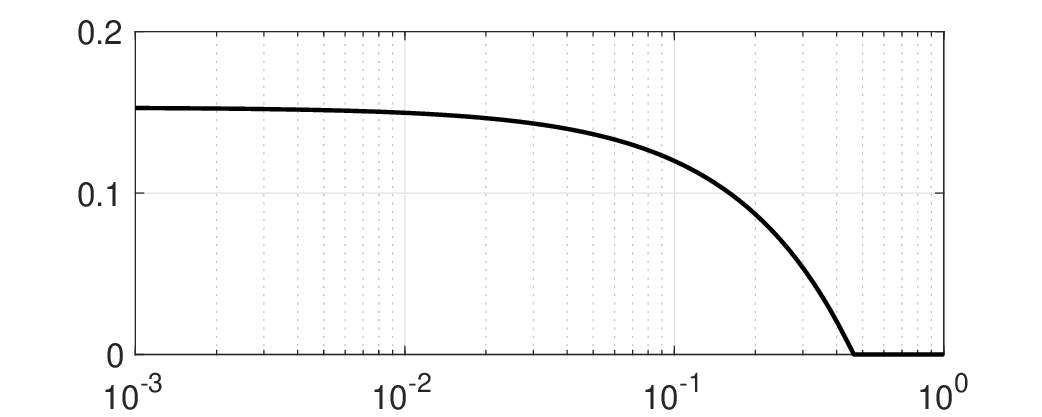}};
	\node[] at (0,-1.9) {\small $d_{\mathrm{TV}}(\tilde{Q}_1,\tilde{Q}_0)=\rho$};
	\node[rotate=90] at (-3.9,0) {\small $ \mathbb{E}\{d_{\mathrm{TV}}(\tilde{P}_S,P_S)\}$};	
	\end{tikzpicture}
	\caption{
		\label{fig:2}
		The upper bound on utility degradation versus the upper bound on disparate impact. The utility improves (degradation decreases) as the disparate impact gets larger. 
	}
\end{figure}

\begin{proof}
	We get
	\begin{align*}
	|\mathbb{E}\{\ell(\mathcal{M}&(\tilde{X}),Y)|S=s\}- \mathbb{E}\{\ell(\mathcal{M}(X),Y)|S=s\}|\\
	&=|\mathbb{E}_{(\tilde{X},Y)\sim \tilde{P}_s}\{\ell(\mathcal{M}(\tilde{X}),Y)\}\\
	&\qquad -\mathbb{E}_{(X,Y)\sim P_s}\{\ell(\mathcal{M}(X),Y)\}|\\
	&=\Bigg|\mathbb{E}_{(\tilde{X},Y)\sim \tilde{P}_s}\{\ell(\mathcal{M}(\tilde{X}),Y)\}
	\\
	&\qquad-\mathbb{E}_{(X,Y)\sim \tilde{P}_s}\Bigg\{\ell(\mathcal{M}(X),Y)\frac{P_s(X,Y)}{\tilde{P}_s(X,Y)}\Bigg\}\Bigg|\\
	&=\Bigg|\mathbb{E}_{(\tilde{X},Y)\sim \tilde{P}_s}\Bigg\{\ell(\mathcal{M}(\tilde{X}),Y)\\
	&\qquad\qquad\qquad\qquad-\ell(\mathcal{M}(\tilde{X}),Y)\frac{P_s(\tilde{X},Y)}{\tilde{P}_s(\tilde{X},Y)}\Bigg\}\Bigg|\\
	&\leq \sigma \Bigg|\mathbb{E}_{(\tilde{X},Y)\sim \tilde{P}_s}\Bigg\{1-\frac{P_s(\tilde{X},Y)}{\tilde{P}_s(\tilde{X},Y)}\Bigg\}\Bigg|\\
	&\leq \sigma\mathbb{E}_{(\tilde{X},Y)\sim \tilde{P}_s}\Bigg\{ \Bigg|1-\frac{P_s(\tilde{X},Y)}{\tilde{P}_s(\tilde{X},Y)}\Bigg|\Bigg\},
	\end{align*}
	where the last inequality follows from the Jensen's inequality~\cite[Theorem~3.9.7]{mukhopadhyay2000probability}.%Finally, note that
\end{proof}

\begin{corollary} The followings hold:
	\begin{align*}
	|\mathbb{P}\{M(\tilde{X})=Y\}-\mathbb{P}\{\mathcal{M}(X)=Y)\}|&\leq 2\mathbb{E}\{d_{\mathrm{TV}}(\tilde{P}_S,P_S)\},\\
	|\mathbb{P}\{M(\tilde{X})\neq Y\}-\mathbb{P}\{\mathcal{M}(X)\neq Y)\}|&\leq 2 \mathbb{E}\{d_{\mathrm{TV}}(\tilde{P}_S,P_S)\}.
	\end{align*}
\end{corollary}

%\begin{proposition}
%Assume that $S-X-\tilde{X}$ forms a Markov chain, i.e., $P_{\tilde{X}|X,S}(\tilde{x}|x,s)=P_{\tilde{X}|X}(\tilde{x}|x)$ for $s\in\{0,1\}$. Then,
%\begin{align*}
%\mathbb{E}\{d_{\mathrm{TV}}(\tilde{P}_S,P_S)\}
%\leq \frac{|\range{Y}|}{2}\sum_{\tilde{x}\in\range{X}} (1-P_{\tilde{X}|X}(\tilde{x}|\tilde{x})).
%\end{align*}
%\end{proposition}
%
%\begin{proof}
%\begin{align*}
%2&d_{\mathrm{TV}}(\tilde{P}_s,P_s)
%\\&=
%\sum_{\tilde{x}\in\range{X},y\in\range{Y}}|\tilde{P}_s(\tilde{x},y)-P_s(\tilde{x},y)|\\
%&=\sum_{\tilde{x}\in\range{X},y\in\range{Y}}\Bigg|\sum_{x\in\mathbb{X}} P_{\tilde{X}|X}(\tilde{x}|x) \mathbb{P}\{X=x,Y=y|S=s\}\\
%&\hspace{.7in} -\mathbb{P}\{X=\tilde{x},Y=y|S=s\}\Bigg|\\
%&\leq \sum_{\tilde{x}\in\range{X},y\in\range{Y}}\Bigg|\sum_{x\in\mathbb{X},x\neq \tilde{x}}\hspace{-.1in} P_{\tilde{X}|X}(\tilde{x}|x) \mathbb{P}\{X=x,Y=y|S=s\}\\
%&\hspace{.7in} +(P_{\tilde{X}|X}(\tilde{x}|\tilde{x})-1)\mathbb{P}\{X=\tilde{x},Y=y|S=s\}\Bigg|\\
%&\leq \sum_{\tilde{x}\in\range{X},y\in\range{Y}} |1-P_{\tilde{X}|X}(\tilde{x}|\tilde{x})|(\mathbb{P}\{X=\tilde{x},Y=y|S=s\}\\
%&\hspace{1.2in}+\max_{x\in\mathbb{X},x\neq \tilde{x}}\mathbb{P}\{X=x,Y=y|S=s\})\\
%&\leq |\range{Y}|\sum_{\tilde{x}\in\range{X}} |1-P_{\tilde{X}|X}(\tilde{x}|\tilde{x})|.
%\end{align*}
%This concludes the proof.
%\end{proof}

\begin{corollary} \label{cor:proof}
Assume that $S-X-\tilde{X}$ forms a Markov chain. Let $\mathcal{P}(\epsilon)$ denote the set of $\epsilon$-differentially private conditional probabilities $P_{\tilde{X}|X}$. Then,
	$\inf_{P_{\tilde{X}|X}\in \mathcal{P}(\epsilon)}  \mathbb{E}\{d_{\mathrm{TV}}(\tilde{P}_S,P_S)\}
	\leq \exp(-\epsilon).$
\end{corollary}

\begin{proof} 
	Define 
	\begin{align*}
	P_{\tilde{X}|X}(\tilde{x}|x)
	=
	\begin{cases}
	\displaystyle 1-\exp(-\epsilon), & \tilde{x}=x,\\[.2em]
	\displaystyle \frac{\exp(-\epsilon)}{(|\range{X}|-1)}, & \tilde{x}\neq x.
	\end{cases}
	\end{align*}
	It can be easily checked that $P_{\tilde{X}|X}(\tilde{x}|x)$ is $\epsilon$ differentially private. We have
	\begin{align*}
	2&d_{\mathrm{TV}}(\tilde{P}_s,P_s)
	\\&=
	\sum_{\tilde{x}\in\range{X},y\in\range{Y}}|\tilde{P}_s(\tilde{x},y)-P_s(\tilde{x},y)|\\
	&=\sum_{\tilde{x}\in\range{X},y\in\range{Y}}\Bigg|\sum_{x\in\mathbb{X}} P_{\tilde{X}|X}(\tilde{x}|x) \mathbb{P}\{X=x,Y=y|S=s\}\\
	&\hspace{.7in} -\mathbb{P}\{X=\tilde{x},Y=y|S=s\}\Bigg|\\
	&= \sum_{\tilde{x}\in\range{X},y\in\range{Y}}\Bigg|\sum_{x\in\mathbb{X},x\neq \tilde{x}}\hspace{-.1in} P_{\tilde{X}|X}(\tilde{x}|x) \mathbb{P}\{X=x,Y=y|S=s\}\\
	&\hspace{.7in} +(P_{\tilde{X}|X}(\tilde{x}|\tilde{x})-1)\mathbb{P}\{X=\tilde{x},Y=y|S=s\}\Bigg|
	\\
	&= \sum_{\tilde{x}\in\range{X},y\in\range{Y}}\Bigg|\sum_{x\in\mathbb{X},x\neq \tilde{x}}\hspace{-.1in} \frac{\exp(-\epsilon)}{(|\range{X}|-1)} \mathbb{P}\{X=x,Y=y|S=s\}\\
	&\hspace{.7in} -\exp(-\epsilon)\mathbb{P}\{X=\tilde{x},Y=y|S=s\}\Bigg|\\
	&= \sum_{\tilde{x}\in\range{X},y\in\range{Y}}\Bigg|\frac{\exp(-\epsilon)}{(|\range{X}|-1)}\Big(\mathbb{P}\{Y=y|S=s\}\\
	&\hspace{1.5in}- \mathbb{P}\{X=\tilde{x},Y=y|S=s\}\Big)\\
	&\hspace{.7in} -\exp(-\epsilon)\mathbb{P}\{X=\tilde{x},Y=y|S=s\}\Bigg|\\
	&= \sum_{\tilde{x}\in\range{X},y\in\range{Y}}\Bigg|\frac{\exp(-\epsilon)}{(|\range{X}|-1)}\mathbb{P}\{Y=y|S=s\}\\
	&\hspace{.7in} -\frac{|\range{X}|}{|\range{X}|-1}\exp(-\epsilon)\mathbb{P}\{X=\tilde{x},Y=y|S=s\}\Bigg|\\
	&\leq \frac{2|\range{X}|}{|\range{X}|-1}\exp(-\epsilon).
	\end{align*}
	This concludes the proof. 
\end{proof}

This corollary establishes the already known result that privacy and utility are conflicting criteria. We cannot achieve one without sacrificing the other. Figure~\ref{fig:1} shows the relationship between the upper bound for the disparate impact [solid] and utility degradation [dashed] versus the differential privacy budget $\epsilon$. Differential privacy is effective in data repair to achieve fairness in small privacy budget regime $\epsilon\ll 1$, albeit at the expense of utility.

\section{Optimal Fairness-Utility Trade-off}
Therefore, the optimal data repair under approximate fairness can be constructed by solving the following optimization problem:
\begin{subequations}\label{eqn:repair_optim_approximate}
	\begin{align} 
	\mathbf{P}_{\rho}:\min_{P_{\tilde{X}|X,S}}&\mathbb{E}\{d_{\mathrm{TV}}(\tilde{P}_S,P_S)\},\\
	\mathrm{s.t.}\;\;  & d_{\mathrm{TV}}(\tilde{Q}_1(\tilde{x}),\tilde{Q}_0(\tilde{x}))\leq \rho.
	\end{align}
\end{subequations}
If we were to require exact fairness in the sense of statistical parity, rather than approximate fairness, the optimal data repair would be given by setting $\rho=0$ in $\mathbf{P}_\rho$. In this case, $\tilde{Q}_0=\tilde{Q}_1=\tilde{Q}$. We would also get
\begin{align*}
d_{\mathrm{TV}}(\tilde{P}_s,P_s)
&=\frac{1}{2}\sum_{\tilde{x}\in\range{X}}\sum_{y\in\range{Y}}|\tilde{P}_s(\tilde{x},y)-P_s(\tilde{x},y)|\\
&\geq \frac{1}{2}\sum_{\tilde{x}\in\range{X}}\left|\sum_{y\in\range{Y}}\tilde{P}_s(\tilde{x},y)-\sum_{y\in\range{Y}}P_s(\tilde{x},y)\right|\\
&=d_{\mathrm{TV}}(\tilde{Q}_s,Q_s),
\end{align*}
where $Q_s(x)=\mathbb{P}\{X=x|S=s\}$. Therefore, the solution to~\eqref{eqn:repair_optim_approximate} is lower bounded by
\begin{subequations}\label{eqn:barycenter}
	\begin{align} 
	\mathbf{P}_{\mathrm{bc}}:\min_{P_{\tilde{X}|X,S}}\pi_0d_{\mathrm{TV}}(\tilde{Q},Q_0)+\pi_1d_{\mathrm{TV}}(\tilde{Q},Q_1),
	\end{align}
\end{subequations}
where $\pi_s=\mathbb{P}\{S=s\}$. Note that $\tilde{Q}$ is the total variation barycenter of $Q_0$ and $Q_1$.

\begin{remark}[Computational Complexity]
The optimization problems in~\eqref{eqn:repair_optim_approximate} and~\eqref{eqn:barycenter} are linear programs with $2|\mathbb{X}|^2$ decision variables. Therefore, the computational complexity of solving them is $\mathcal{O}(|\mathbb{X}|^6)$ using the interior point method~\cite{POTRA2000281}. 
\end{remark}

\begin{proposition}
Let $\mathcal{S}_\rho$ denote the set of solutions of $\mathbf{P}_\rho$ in~\eqref{eqn:repair_optim_approximate}. Then, $\lim_{\rho\rightarrow 0}\mathrm{dist}(\mathbf{P}_\rho,\mathbf{P}_0)=0$.
\end{proposition}

\begin{proof}
The proof follows from~\cite[Theorem~2.4]{mangasarian1987lipschitz}.
\end{proof}

\begin{proposition} \label{prop:decreasing}
Let $\mathcal{S}_\rho$ denote the set of solutions of $\mathbf{P}_\rho$ in~\eqref{eqn:repair_optim_approximate}.  Then, $\mathbb{E}\{d_{\mathrm{TV}}(\tilde{P}_S,P_S)\}$ at the optimal solution $P_{\tilde{X}|X,S}\in\mathcal{S}_{\rho}$ is decreasing in $\rho$. 
\end{proposition}

\begin{proof}
The proof follows from that the constraint set shrinks as $\rho$ gets smaller, i.e., $\{P_{\tilde{X}|X,S}\,|\,d_{\mathrm{TV}}(\tilde{Q}_1(\tilde{x}),\tilde{Q}(\tilde{x}))\leq \rho_1\}\subseteq \{P_{\tilde{X}|X,S}\,|\,d_{\mathrm{TV}}(\tilde{Q}_1(\tilde{x}),\tilde{Q}(\tilde{x}))\leq \rho_2\}$ if $\rho_1\leq \rho_2$. 
\end{proof}

Proposition~\ref{prop:decreasing} shows that the solution to~\eqref{eqn:repair_optim_approximate} is lower bounded by~\eqref{eqn:barycenter} for, in fact, any $\rho$. This illustrates the importance of the  total variation barycenter in fairness. 

\section{Numerical Experiments}
This dataset is composed of nearly 49,000 anonymized records from the 1994 Census database~\cite{Dua:2017}. The objective of the forecasting model is to use inputs, such as age, education, and work, to determine whether an individual earns more than \$50,000. This is a classification task. We focus on two inputs: work per week (hours) and education (years). The sensitive or protected attribute is gender: male and female. Cultural factors can contribute to disparities on work hours per week between male and female individuals or their education level. These cultural norms should not be reflected in forecasting models. %However, this should not change the ability of an individual earning more than \$50,000. 

Figure~\ref{fig:3} [left] shows the histogram of work per week and education for male and female individuals for the original data. The effect of gender divide is clear on work hours per week; male individuals tend to work longer hours. We can pre-process the data to achieve fairness. Figure~\ref{fig:3} [middle] illustrates the histogram of work per week and education for male and female individuals for pre-processed data with approximate fairness with $\rho=0.1$. The data for male and female now looks very similar for longer hours while shorter hours are still dominated by female individuals. Figure~\ref{fig:3} [right] illustrates the histogram of work per week and education for male and female individuals for pre-processed data to ensure statistical parity. In this case, the statistics are identical for male and female individuals. 

Now, we can investigate the trade-off between fairness and utility. Figure~\ref{fig:2} illustrates The upper bound on utility degradation versus the upper bound on disparate impact. The utility improves (degradation decreases) as the disparate impact gets larger. 

\section{Conclusions}
We used total variation distance to capture fairness and utility in pre-processing for fair machine learning. We were able to cast the problem of finding the optimal pre-processing regiment for enforcing fairness as minimizing total variations distance between the distribution of the data before and after pre-processing subject to a constraint on the total variation distance between the distribution of the inputs given protected attributes. This problem is a linear program that can be efficiently solved. We demonstrated the results using numerical experimentation on a practice dataset.

\bibliography{citation}
\bibliographystyle{ieeetr}

\end{document}